\documentclass[twoside,12pt,reqno]{amsart}
\usepackage{amsmath,amsthm, amsfonts,amssymb}
\usepackage{color,soul}
\usepackage[dvipsnames]{xcolor}
\usepackage[all]{xy}
\usepackage{graphicx}
\usepackage{indentfirst}
\usepackage{bm}
\usepackage{mathrsfs}
\usepackage{latexsym}
\usepackage{hyperref}
\usepackage{microtype}	
\usepackage{bookmark}
\usepackage{tikz}[2010/10/13]
\input{Qcircuit}
\usepackage{mathtools}
\usepackage{breakurl}
\usepackage{ifpdf}
\usepackage{doi}

\newtheorem{Thm}{Theorem}[section]

\newtheorem{proposition}[Thm]{\bf Proposition}

\newtheorem{definition}[Thm]{\bf Definition}

\newtheorem{theorem}[Thm]{\bf Theorem}

\newcommand{\abs}[1]{\left\vert #1 \right\vert}

\newcommand{\be}{\begin{equation}}
\newcommand{\ee}{\end{equation}}
\newcommand{\blue}[1]{{\color{blue}#1}}

\newcommand{\beq}{\begin{eqnarray}}
\newcommand{\eeq}{\end{eqnarray}}

\newcommand{\Z}{\mathbb{Z}}
\newcommand{\beqs}{\begin{eqnarray*}}
\newcommand{\eeqs}{\end{eqnarray*}}
\newcommand{\Max}{\ket{\rm{Max}}}
\newcommand{\GHZ}{\ket{\rm{GHZ}}}

\newcommand{\FS}{\mathfrak{F}_{s}}

\newcommand{\size}[1]{\fontsize{10pt}{\baselineskip}\selectfont{#1}}
\renewcommand{\leq}{\leqslant}

\newcommand{\fmeasure}[5]{
\node at (#1-#3,#2-#4) {\size{$#5$}};
\draw (#1,#2) --(#1,#2-#4)  arc (0:180:#3) -- (#1-#3-#3,#2);
}

\newcommand{\fqudit}[5]{
\node at (#1-#3,#2+#4) {\size{$#5$}};
\draw (#1,#2) --(#1,#2+#4)  arc (0:-180:#3) -- (#1-#3-#3,#2);
}

\newcommand{\fdoublequdit}[6]
{
\fqudit{#1}{#2}{3*#3}{#4}{}
\fqudit{#1-2*#3}{#2}{#3}{#4}{#5}
\node at (#1-3*#3, #2+#4+2*#3) {\size{$#6$}};
}

\newcommand{\fbraid}[4]{
\draw (#1,#2)--(#3,#4);
\draw (#1,#4)--(2/3*#1+1/3*#3,2/3*#4+1/3*#2);
\draw (#3,#2)--(2/3*#3+1/3*#1,2/3*#2+1/3*#4);
}

\title[Constructive Simulation and Topological Design]{%\blue{Constructive Simulation}\\
 Constructive Simulation and \\Topological Design of Protocols}
\author{Arthur Jaffe}
\address{Harvard University, Cambridge, MA 02138, USA \\ and Max Planck Institute for Mathematics, Bonn, Germany}
\email{arthur\_jaffe@harvard.edu}
\author{Zhengwei Liu}
\address{Harvard University, Cambridge, MA 02138, USA}
\email{zhengweiliu@fas.harvard.edu}
\author{Alex Wozniakowski}
\address{Harvard University, Cambridge, MA 02138, USA}
\email{airwozz@gmail.com}

\begin{document}
\begin{abstract}
We give a topological simulation for tensor networks that we call the {\it two-string model}.  In this approach we give a new way to design protocols, and we discover a new multipartite quantum communication protocol.   We  introduce the notion of topologically-compressed transformations.   Our new protocol can implement multiple, non-local compressed transformations among multi-parties using one multipartite resource state. 
\end{abstract}
\maketitle

\section{Constructive Simulation and Topological Design}
By  \textit{constructive simulation} we mean the development of a  pictorial language for quantum information that yields intuition and insight, as well as understanding.  Just as the choice of language can  determine style or tone in writing, the choice of a mathematical language can influence one's pattern of thinking. Different languages convey different ideas and insights related to the same content.   A good language can suggest the discovery of new relations and aid the invention of new concepts.  
  
Manin and Feynman introduced the concept of quantum simulation \cite{Manin-book,Feynman,Manin},
and here we explore the simulation of quantum processes from the point of view of using a  diagrammatic language.
We are not the first to study topological methods; early landmark papers using topological methods  in quantum information  include 
 \cite{Kitaev03,Freedman-etal,LW}.  
The categorical approach  has been studied extensively in quantum information and tensor networks  by many persons \cite{DVC,AbramskyCoecke04,Coecke-Duncan-2,Coeckebook,BB,DBJC,LimitationSanders,Backens}. People also studied  these subjects  from a  planar algebra point of view~\cite{Vicary,Vicary-Reutter}.  We believe that our topological simulation,  presented in this paper, is significantly different from past attempts and will prove to be useful.

Pictures have for a long time complemented algebra as a way to provide  guidance. 
%We utilize our diagrammatic language and its rules to recover some fundamental concepts of quantum information and tensor networks in a natural way. 
We focus on communication which is intrinsic to quantum networks: this is the task of propagating information from one place in the network to another.
It is reasonable to think that topological simulation based on isotopy is sufficient. In fact, communication seems especially suited for topological design, as quantum communication protocols can be expressed in purely topological form. We showed in ~\cite{JLW} that one can recover fundamental concepts in quantum information in this way.  In this paper we show how to define new concepts and to use topological simulation to design new protocols.

In \S\ref{Sect:TopologicalCompression}--\S\ref{Sect:CompressionDef} we explain the concept of topological simulation in detail, and we define \textit{topologically-compressed transformations}---a category of transformations that includes all controlled transformations.  

Then in \S\ref{Sect:MCT} we use topological {\it para-isotopy} to design a new diagrammatic protocol that we call {\it multipartite compressed teleportation} (MCT).  We apply this protocol to implement multiple non-local compressed transformations among multi parties---using one entangled state as a resource state,  local transformations, and classical communication~(LOCC).
We  show how one can represent MCT in terms of the usual algebraic elements that one employs in circuit design.   This protocol improves the efficiency of teleportation, compared with two-party communication, by a factor of two.

The concepts of constructive simulation and topological design are model independent.  In this paper we study what we call the \textit{two-string model}.  In this model we can simulate the Pauli matrices, measurement, and the resource state in a topological way.  In many communication protocols, the measurement-based recovery map is given by Pauli matrices. 
The resource state is the Bell state, or the GHZ state \cite{GHZ}.
Our model provides a topological explanation of this fact.  

It is interesting to find if other protocols, such as factoring \cite{Shor} or secure sharing \cite{Lance-etal}, require using other  elements of simulation in addition to topology.

\section{Fundamental Diagrams in the Two-String Language}
Our \textit{two-string language} acquires its name from the fact that we represent transformations of $1$-qudits by  diagrams with two input points and two output points.  We obtain fundamental diagrams for resource states, measurements,   the Pauli matrices $X,Y,Z$,
 %the Fourier transform $F$, the Gaussian $G$ (that implements a braid),
 and the string Fourier transform $\mathfrak{F}_{S}$. 
 Our strings are charged, with a label indicating the charge $k\in\Z_{d}$; this means that we consider charges modulo $d$.

The reader may wish to read complete details about  the two-string language that we present in~\cite{JL,JLW}.  However, in order to make this paper self-contained, we explain in this section  those aspects of the language that we require in this paper---without repeating the detailed proofs.

\subsection{Qudits and Transformations}
Let $d$ be the dimension of the single qudit space. We represent qudits by charged strings in the shape of a cap.   We generally omit the label for any charge $k_{j}=0$.  
We place our strings in the plane.  

Our convention is to place the charge on the left side of a vertical string.  
Isotopies that reverse this placement  are not allowed. 
However the string-Fourier relation  allows one to move a charge label over a cap or under a cup, see \eqref{Equ:SF1}.  One can use this relation to enable isotopies that would otherwise move a label  across a string from one side to the other.  

We represent the  $n$-qudit basis $\vec{\ket{k}}=\ket{k_1,k_2,\cdots, k_n}$ by $n$ charged caps, and these have $2n$ output points:
\be
\vec{\ket{k}}
=\frac{1}{d^{n/4}}\ \raisebox{-.5cm}{\tikz{
\fqudit{0/-3}{0/3}{1/-3}{2/3}{k_2}
\fqudit{4/-3}{0/3}{1/-3}{3/3}{k_1}
\fqudit{-6/-3}{0/3}{1/-3}{1/3}{k_n}
\node at (-4/-3,1/3) {$\cdots$};
}}\quad.
\ee
By convention, we place the label on the righthand string in each cap.

One denotes the adjoint by a charge-inverting vertical reflection, so the $n$-qudit matrix units $\vec{\ket{k}}\vec{\bra{\ell}}=\ket{k_1,k_2,\cdots, k_n}\bra{\ell_1,\ell_2,\cdots, \ell_n}$ are represented by:
\be
\vec{\ket{k}}\vec{\bra{\ell}}
=\frac{1}{d^{n/2}}
\raisebox{-1.2cm}{
\tikz{
\fqudit{0/-3}{0/3}{1/-3}{2/3}{k_2}
\fqudit{4/-3}{0/3}{1/-3}{3/3}{k_1}
\fqudit{-6/-3}{0/3}{1/-3}{1/3}{k_n}
\node at (-4/-3,7/3) {$\cdots$};
\fmeasure{0/-3}{9/3}{1/-3}{2/3}{-\ell_2}
\fmeasure{4/-3}{9/3}{1/-3}{3/3}{-\ell_1}
\fmeasure{-6/-3}{9/3}{1/-3}{1/3}{-\ell_n}
\node at (-4/-3,1/3) {$\cdots$};
}}\quad.
\ee
Therefore any $n$-qudit transformation $T$ is a diagram with $2n$ input points on the top and $2n$ output points on the bottom,
\be
T=\underbrace{\raisebox{-.4cm}{
\tikz{
\draw (-1/6,1/3) rectangle (1/6+2/3,1-1/3);
\draw (0,0)--(0,1/3);
\draw (2/3,0)--(2/3,1/3);
\draw (0,1)--(0,1-1/3);
\draw (2/3,1)--(2/3,1-1/3);
\node at (1/3,1/2) {$T$};
\node at (1/3,1/6) {$\cdots$};
\node at (1/3,1-1/6) {$\cdots$};
}}
}_{2n} \;.
\ee

\subsection{Planar relations}\label{Sec:planar relation}
In this section we give relations between certain diagrams; the consistency of these relations is proved in \cite{JL}. These relations provide a dictionary that relates qudits, transformations, measurements, and diagrams. It is crucial that any diagram with $2n$ input points and $2n$ out points represents an $n$-qudit transformation.

Recall that $d$ is the dimension of the $1$-qudit space. Let $q=e^{\frac{2\pi i}{d}}$, and  $\zeta=q^{1/2}$ be a square root of $q$ satisfying  $\zeta^{d^{2}}=1$.

\subsubsection{Multiplication yields additive charge of order $d$}
\be\label{AddCharge}
\raisebox{-.5cm}{
\begin{tikzpicture}
\draw (0,0) --(0,1);
\node (0,0) at (-0.2,0.8) {$\ell$};
\node (0,0) at (-0.2,0.2) {$k$};
\node (0.45,0.25) at (0.45,0.5) {$=$};
\node (0.9,0.35) at (1.4,0.5) {$k + \ell$};
\draw (2,0) --(2,1);
\node (1.4,-0.05) at (2.3,0.3) {,};
\end{tikzpicture} }\qquad
\raisebox{-.5cm}{
\tikz{
\node (3.7,0.35) at (4.8,0.5) {$d$};
\draw (5,0) --(5,1);
\node (4.4,0.25) at (5.4,0.5) {$=$};
\draw (5.8,0) --(5.8,1);
\node (4.95,0) at (6.1,0.3) {.};
}}
\ee

\subsubsection{Para-isotopy for exchange of charge order}
\be \label{Equ:para isotopy}
\raisebox{-.5cm}{
\begin{tikzpicture}
\draw (0,0) --(0,01);
\draw (0.2,0) --(0.2,01);
\draw [fill] (0.4,0.5) circle [radius=0.01];
\draw [fill] (0.5,0.5) circle [radius=0.01];
\draw [fill] (0.6,0.5) circle [radius=0.01];
\node (0,0) at (-0.15,0.2) {$k$};
\draw (0.8,0) --(0.8,1);
\draw (1.15,0) --(1.15,1);
\node (1.2,0) at (1,0.8) {$\ell$};
\end{tikzpicture}
}
=q^{k\ell}
\raisebox{-.4cm}{
\begin{tikzpicture}
\draw (0,0) --(0,01);
\draw (0.2,0) --(0.2,01);
\draw [fill] (0.4,0.5) circle [radius=0.01];
\draw [fill] (0.5,0.5) circle [radius=0.01];
\draw [fill] (0.6,0.5) circle [radius=0.01];
\node (0,0) at (-0.15,0.8) {$k$};
\draw (0.8,0) --(0.8,1);
\draw (1.15,0) --(1.15,1);
\node (1.2,0) at (1,0.2) {$\ell$};
\end{tikzpicture}}\;.
\ee
Here we assume that the strings between charge-$k$ string  and charge-$\ell$ string are neutral. 

\subsubsection{Twisted tensor product}
The twisted tensor product interpolates between the two vertical orders of the product.  In the twisted product, we write the labels at the same vertical height:
\beq\label{TwistedProduct}
\raisebox{-.4cm}{
\begin{tikzpicture}
\draw (0,0) --(0,01);
\draw (0.2,0) --(0.2,01);
\draw [fill] (0.4,0.5) circle [radius=0.01];
\draw [fill] (0.5,0.5) circle [radius=0.01];
\draw [fill] (0.6,0.5) circle [radius=0.01];
\node (0,0) at (-0.15,0.5) {$\scriptstyle{k}$};
\draw (0.8,0) --(0.8,1);
\draw (1.15,0) --(1.15,1);
\node (1.2,0) at (1,0.5) {$\scriptstyle{\ell}$};
\end{tikzpicture}}
&\equiv&\zeta^{-k\ell}
\raisebox{-.4cm}{
\begin{tikzpicture}
\draw (0,0) --(0,01);
\draw (0.2,0) --(0.2,01);
\draw [fill] (0.4,0.5) circle [radius=0.01];
\draw [fill] (0.5,0.5) circle [radius=0.01];
\draw [fill] (0.6,0.5) circle [radius=0.01];
\node (0,0) at (-0.15,0.2) {$\scriptstyle{k}$};
\draw (0.8,0) --(0.8,1);
\draw (1.15,0) --(1.15,1);
\node (1.2,0) at (1,0.8) {$\scriptstyle{\ell}$};
\end{tikzpicture}}
\nonumber\\
&=&\zeta^{k\ell}
\raisebox{-.4cm}{
\begin{tikzpicture}
\draw (0,0) --(0,01);
\draw (0.2,0) --(0.2,01);
\draw [fill] (0.4,0.5) circle [radius=0.01];
\draw [fill] (0.5,0.5) circle [radius=0.01];
\draw [fill] (0.6,0.5) circle [radius=0.01];
\node (0,0) at (-0.15,0.8) {$\scriptstyle{k}$};
\draw (0.8,0) --(0.8,1);
\draw (1.15,0) --(1.15,1);
\node (1.2,0) at (1,0.2) {$\scriptstyle{\ell}$};
\end{tikzpicture}
}\;.
\eeq
In this case $k,l\in \Z$, and $k$ and $k+d$ yield  different diagrams. If the pair is neutral, namely $\ell=-k$, then the twisted tensor product is defined for $k\in \Z_d$. 

\subsubsection{The string Fourier relation, for moving charge across a cap or cup}
\label{Sect:String-Fourier}
\be \label{Equ:SF1}
   \raisebox{-0.2cm}{
   \tikz{
   \fqudit {0}{0}{1/-3}{1/3}{}
   \node at (-3/-3,1/3) {\size{$\hskip -2.5cm k$}};
   }}
    =\zeta^{k^2}~~
    \raisebox{-0.2cm}{\tikz{\fqudit {0}{0}{1/-3}{1/3}{k}}}
    \;,\qquad\text{and}\qquad
   \raisebox{-0.2cm}{
   \tikz{
   \fmeasure {0}{0}{1/-3}{1/3}{};
   \node at (-3/-3,-1/3) {\size{$\hskip -2.5cm k$}};
   }}
    =\zeta^{-k^2}
    \raisebox{-0.2cm}{\tikz{\fmeasure {0}{0}{1/-3}{1/3}{k}}}
    \;.
\ee

\subsubsection{Quantum dimension}

\be\label{QuantumDimension}
\raisebox{-.4cm}{
\begin{tikzpicture}
\draw (0.8,0) circle [radius=0.5];
\end{tikzpicture}}
= \sqrt{d}
\;.
\ee

\subsubsection{Neutrality}%:
\be\label{Neutrality}
\raisebox{-.4cm}{
\begin{tikzpicture}
\node (0,0) at (0.45,0.05) {$k$};
\draw (1.1,0) circle [radius=0.5];
\end{tikzpicture}}
=0\;,\quad \text{for } d\nmid k.
\ee

\subsubsection{Temperley-Lieb relation}%:
\be
\raisebox{-.7cm}{
\tikz{
\draw (0,0)--(0,1/3+1/6) arc (180:0:1/6) arc (180:360:1/6) -- (2/3,1);
\draw (2/3+0,0+1/2)--(2/3+0,1/3+1/6+1/2) arc (0:180:1/6) arc (0:-180:1/6) -- (2/3+-2/3,1+1/2);
}}\
=
\raisebox{-.7cm}{
\tikz{
\draw (0,0)--(0,1+1/2);
}}\;,\qquad
\raisebox{-.7cm}{
\tikz{
\draw (0,0)--(0,1/3+1/6) arc (0:180:1/6) arc (0:-180:1/6) -- (-2/3,1);
\draw (-2/3,1/2)--(-2/3,1/3+1/6+1/2) arc (180:0:1/6) arc (180:360:1/6) -- (0,1+1/2);
}}\
=
\raisebox{-.7cm}{
\tikz{
\draw (0,0)--(0,1+1/2);
}}\; .
\ee
Based on the Temperley-Lieb relation, a neutral string only depends on the end points:
\be
\raisebox{-.5cm}{
\tikz{
\draw (0,0)--(0,1/3+1/6) arc (180:0:1/6) arc (180:360:1/6) -- (2/3,1);
}}\
=
\raisebox{-.5cm}{
\tikz{
\draw (0,0)--(2/3,1);
}}\;,\qquad
\raisebox{-.5cm}{
\tikz{
\draw (0,0)--(0,1/3+1/6) arc (0:180:1/6) arc (0:-180:1/6) -- (-2/3,1);
}}\
=
\raisebox{-.5cm}{
\tikz{
\draw (0,0)--(-2/3,1);
}}\; .
\ee

\subsubsection{Resolution of the identity}%:
\begin{align}\label{Equ:Resolution of the identity}
\raisebox{-.5cm}{
\tikz{
\draw (0,0)--(0,1+1/6);
\draw (2/3,0)--(2/3,1+1/6);
}}
&= d^{-1/2}\sum_{k=0}^{d-1}
\raisebox{-.5cm}{
\tikz{
\fqudit {0}{-3.5/3}{1/-3}{.5/3}{k}
\fmeasure {0}{0}{1/-3}{.5/3}{-k}
}}\quad .
\end{align}

\subsubsection{Braid}
We show in Proposition 2.15 of \cite{JL}  that 
\be
\omega=\frac{1}{\sqrt{d}}\sum_{j=0}^{d-1}\zeta^{j^{2}} 
\ \text{satisfies} \ \abs{\omega}=1\;.
\ee
The positive braid is 
\beq\label{first-pos-braid}
\raisebox{-.3cm}{
\tikz{
\fbraid{0}{2/3}{2/-3}{0}
}}
&\ \equiv&
\frac{1}{\sqrt{\omega d}}\,
\sum_{k=0}^{d-1}
\raisebox{-.3cm}{
\tikz{
\draw (0/-3,0) --(0/-3,3/3);
\draw (2/-3,0) --(2/-3,3/3);
\node at (1/-3,1/3) {$-k$};
\node at (3/-3,2/3) {$k$};
}}\\
&=&
\frac{1}{\sqrt{\omega d}}\,
\sum_{k=0}^{d-1} \zeta^{k^{2}}\,
\raisebox{-.3cm}{
\tikz{
\draw (0/-3,0) --(0/-3,3/3);
\draw (2/-3,0) --(2/-3,3/3);
\node at (1/-3,1.5/3) {$-k$};
\node at (3/-3,1.5/3) {$k$};
}}\quad.\nonumber
\eeq

\subsubsection{Pauli matrices}
The Pauli matrices $X,Y,Z$ are
\be\label{Pauli 21}
I =
\raisebox{-.4cm}{
\tikz{
\draw (0,0)--(0,1);
\draw (2/3,0)--(2/3,1);
}}\;,\quad
X=
\raisebox{-.4cm}{
\tikz{
\draw (0,0)--(0,1);
\draw (2/3,0)--(2/3,1);
\node at (-1/3+2/3,1/2) {1};
}}\;,\quad
Y=
\raisebox{-.4cm}{
\tikz{
\draw (0,0)--(0,1);
\draw (2/3,0)--(2/3,1);
\node at (-1/3,1/2) {-1};
}}\;,\quad
Z=
\raisebox{-.4cm}{
\tikz{
\draw (0,0)--(0,1);
\draw (2/3,0)--(2/3,1);
\node at (-1/3,1/2) {1};
\node at (-1/3+2/3,1/2) {-1};
}}\;.
\ee

\subsubsection{Bell State}
Moreover the Bell state, as a two-qudit resource state shared by Alice and Bob, is $ d^{-1/2}$ times
\be\label{Bell-State}
\raisebox{-.5cm}{
\tikz{
\fdoublequdit{1/-3}{0}{.5/-3}{.5/3}{}{}
\draw [red,dashed] (-.5/-3,0/3)-- (-.5/-3,4/3);
\node at (-3/-3,3/3) {Alice};
\node at (2/-3,3/3) {Bob};
}} 
\ee
Here only the double caps represent the Bell state; the other labels are for explanation. The dashed, red line indicates that the two persons have distinct localizations. The double cap can pass the red line. This means that the corresponding state can be shared between Alice and Bob as an entangled resource state. We use a corresponding $n$-qudit resource state given in \eqref{Max-State} for designing our protocol.

\subsection{Algebraic relations for some 1-qudit transformations}
The algebraic definitions of the Pauli $X,Y,Z$ are 
	\be\label{XYZ-Defn-1}
	X\ket{k}=\ket{k+1}\;,\
	Y\ket{k}=\zeta^{1-2k}\ket{k-1}\;,\
	Z\ket{k}=q^{k}\ket{k}\;.	
	\ee
The quantum Fourier transform matrix $F$ and the Gaussian matrix $G$ are defined by 
	\be\label{F-and-G}
	F\ket{k}=\frac{1}{\sqrt{d}}\sum_{\ell=0}^{d-1} q^{k\ell} \ket{\ell}\;,
\quad
G\ket{k}=\zeta^{k^2}\ket{k}\;.
	\ee
These matrices are unitary, and they satisfy many interesting relations, including 
	\be
	X^{d} = Y^{d} = Z^{d} =F^{4}=G^{2d}=(FG)^{3}\omega^{-1}=I\;,
	\ee 
	\be\label{XY-qYX}
		XYX^{-1}Y^{-1}
		= YZY^{-1}Z^{-1}
		= ZXZ^{-1}X^{-1}
		=q \;,
	\ee
	\be\label{XYZ-Equivalence}
	XYZ=\zeta\;,\quad
	FXF^{-1}=Z\;,\quad
	GXG^{-1}=Y^{-1}\;.
	\ee

We have shown that these matrices generate the single qudit Clifford group in Proposition~9.1 of~\cite{JL}.
The representations of the 1-qubit Clifford group were studied in \cite{WF}. 
For the multi-qubit case, the representations of the Clifford group were studied in~\cite{HWW}, where one can find  further references about the applications in quantum information.  It would be interesting to generalize those results to the qudit case. 

\subsection{Multipartite Resource States}
Greenberger, Horne, and Zeilinger introduced the classic multipartite resource state that we denote $\GHZ$ in~\cite{GHZ}.  Experimental work on $\GHZ$ was achieved in \cite{Pan-3,Experimental,Pan-4}.  

We introduced our $n$-qudit resource state $\Max$ in \cite{JLW}.  This state generalizes the Bell state \eqref{Bell-State} and has the diagrammatic representation,  
	\be\label{Max-State}
	\Max = d^{-n/4}\,
\underbrace{
\raisebox{-.5cm}{
\tikz{
\fqudit{0/-3}{0/3}{1/-3}{1/3}{}{}
\fqudit{-4/-3}{0/3}{1/-3}{1/3}{}{}
\node at (-3/-3,1/3) {$\cdots$};
\fqudit{1/-3}{0/3}{4/-3}{1/3}{}{}
}}}_{2n}\;.
	\ee

The algebraic interpretation of this resource state is also interesting.  
Let $ \vec{\ket{k}}=\ket{k_{1},k_{2}, \ldots, k_{n}}$ denote an $n$-qudit state with charges $\vec{k}=(k_{1},\cdots,k_{n})$.  Also let 
 $|\vec{k}|=\sum_{j=1}^{n}k_{j}\in \Z_{d}$ denote the total charge.  We have shown that 
\be\label{GHZ}
\GHZ=d^{-\frac{1}{2}}\sum_{k=0}^{d-1}\ket{k,k,\cdots,k} \;,
\quad\text{while}\quad 
\Max=d^{\frac{1-n}{2}}\sum_{|\vec{k}|=0} \ket{k_{1},k_{2}, \ldots, k_{n}}\;.
\ee
In fact these two resource states are related by the local operation of the quantum Fourier transform, 
\be\label{GHZ-Max}
\GHZ=(F\otimes \cdots \otimes F)\Max\;,
\ee
where $F$ denotes the  $1$-qudit quantum Fourier transform defined in \eqref{F-and-G}. 

\section{Teleportation}
One could say that the modern theory of quantum communication networks began in 1993 with the teleportation protocol discovered by  Bennett, Brassard, Cr\'epeau, Jozsa, Peres, and Wootters    \cite{Bennett-etal}.  This  protocol  allows one to disassemble a quantum state located at Alice's location, and to reconstruct it at Bob's location.   In order for the reconstruction to succeed, Alice and Bob prearrange to share a specific entangled state, which is utilized as a resource for the protocol.  In addition, they share some purely classical information.

Preskill, Gottesman, and Chuang described the notion of quantum software for solving problems in quantum computation and quantum communication \cite{Preskill,GottesmanChuang}.  Recently, Pirandola and Braunstein cite teleportation as the  ``most promising mechanism for a future quantum internet'' \cite{PirandolaBraunstein}.
One can realize quantum networks through bidirectional quantum teleportation (BQST).
Experimental work on long-distance teleportation has been achieved \cite{Pan-ea,Zeilinger-ea,Valivarthi-ea}.  The \textit{Quantum Science Satellite} built by Pan and his coworkers provides an opportunity to  test teleportation at record-breaking distances~\cite{Pan-1,Pan-2}.

 \section{Topological Compression: Informal Discussion}\label{Sect:TopologicalCompression}
 A fundamental concept that we introduce in this paper is a \textit{topologically-compressed transformation}.   
Our notion of topological compression becomes transparent in terms of the two-string model for quantum information.  One can visualize compression of a transformation in terms of the diagrams that describe it. 
Topological compression is compatible with use of our  multipartite resource state $\Max$ illustrated in \eqref{Max-State}.

Basically, the information for a compressed transformation on a qudit is carried by one of two strings.   For transformations on a single qudit, the Pauli matrices $X,Y$ in the representation \eqref{Pauli 21} are compressed. But Pauli $Z$ is unitarily equivalent to Pauli $X$, so it too is compressed.  Let us explain this in terms of a more general  example.  

\color{black}
Suppose that Alice and Bob are at separate locations and want to implement a non-local, two-qudit transformation~$T$.  The topological simulation of that goal is given by the following diagram:
$$
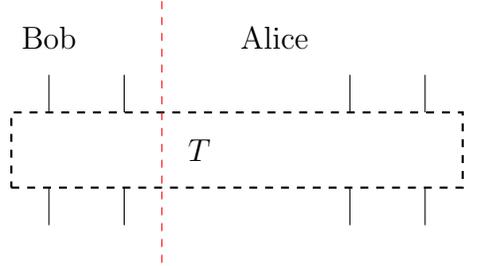
\tikz{
\node at (3,2.5) {Alice};
\node at (0,2.5) {Bob};
\foreach \x in {0,1,4,5}{
\draw (\x,0)--++(0,2);
}
\fill[white] (-0.5,.5) rectangle (5.5,1.5);
\draw[thick,dashed] (-0.5,.5) rectangle (5.5,1.5);
\node at (2,1) {$T$};
\draw[red,dashed] (1.5,-.5)--(1.5,3);
}\raisebox{1.3cm}{\quad  .}
$$
If one applies a topological isotopy, one can move the red line so the transformation is performed completely by Alice. This is the solution, and its topological simulation is 
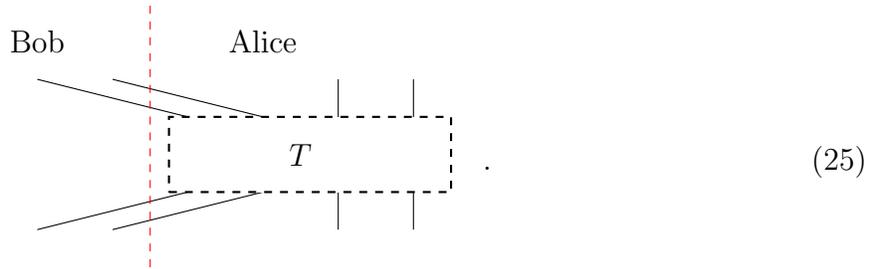
\be\label{Isotopy-Teleport}
\raisebox{-1.3cm}{\tikz{
\node at (3,2.5) {Alice};
\node at (0,2.5) {Bob};
\foreach \x in {4,5}{
\draw (\x,0)--++(0,2);
}
\foreach \x in {0,1}{
\draw (\x,0)--++(2,.5)--++(0,1)--++(-2,.5);
}
\fill[white] (1.75,.5) rectangle (5.5,1.5);
\draw[thick,dashed] (1.75,.5) rectangle (5.5,1.5);
\node at (3.5,1) {$T$};
\draw[red,dashed] (1.5,-.5)--(1.5,3);
}}\quad  .
\ee
This topological isotopy does not change the function of the diagram, but does change its interpretation in quantum information. After isotopy, the diagram means that Bob  can teleport his input to Alice; then Alice can implement the transformation $T$ locally on her computer and teleport the result back to Bob using BQST.

It is well-known that the cost of teleportation for a general transformation is two resource states.   
Recall that one resource state allows two strings to pass across the red dashed line, as explained for the Bell state in \eqref{Bell-State}.  Thus one has an indication from \eqref{Isotopy-Teleport}  that the cost of teleportation can be estimated by counting the number of strings that pass over the red dashed line.  

However Zhou et al. and Eisert et al.
pointed out that the cost of BQST may not be optimal. For certain transformations, including CNOT, they gave a teleportation protocol with lower cost  \cite{Zhou-etal, Eisert-etal}.  This optimization has been further studied in \cite{Yu-etal}, and in \cite{Pirandola-etal,Yu} one finds extensive references.

So it is natural to ask the question: what transformations can be teleported with less cost, compared with BQST?   We now characterize  topologically-compressed transformations and show that they have this property.  All controlled transformations are topologically compressed.  

Following the above discussion, consider any 2-qudit transformation that can be represented by the following diagram: 
$$\tikz{
\foreach \x in {0,1,2,3}{
\draw (\x,0)--++(0,1);
}
\fill[white] (0.5,.25) rectangle (3.5,.75);
\draw[thick,dashed] (0.5,.25) rectangle (3.5,.75);
\node at (2,.5) {$T$};
}$$
In other words, such a transformation acting on Bob's qudit only requires the information on one of the two strings.  We say that such transformations are topologically compressed; we give the precise statement as Definition \ref{Def:Compress}. This covers a large variety of transformations that are commonly used in protocols.
In this case the corresponding isotopy yields: 
$$\tikz{
\node at (3,2.5) {Alice};
\node at (0,2.5) {Bob};
\foreach \x in {4,5}{
\draw (\x,0)--++(0,2);
}
\draw (0,0)--(0,2);
\foreach \x in {1}{
\draw (\x,0)--++(2,.5)--++(0,1)--++(-2,.5);
}
\fill[white] (1.75,.5) rectangle (5.5,1.5);
\draw[thick,dashed] (1.75,.5) rectangle (5.5,1.5);
\node at (3.5,1) {$T$};
\draw[red,dashed] (1.5,-.5)--(1.5,3);
}.
$$
As only two strings pass over the red line, one expects that it is possible to implement this non-local transformation using only one resource state, rather than two.

In  \S\ref{Sect:MCT}, we introduce a new protocol to teleport information that is captured on one string using one resource state. We call this protocol {\it compressed teleportation} (CT). It relies on using local operations and classical communication (LOCC).  Furthermore it optimizes the entanglement resource cost for teleportation of compressed transformations. The CT protocol reduces the costs by 50$\%$ compared with BQST.  In BQST, one needs two resource states. In our CT protocol, we use one resource state.
Since  CNOT, Tofolli, and controlled transformations are all topologically compressed, our protocol covers the previous teleportation protocols for CNOT, Tofolli, and controlled transformations.

We then generalize this protocol to multipartite communication in \S\ref{Sect:MCT-Control}.  
Our MCT protocol does not reduce to multiple, bipartite communications.  
If one realizes this teleportation by BQST, then one would need $n$ bipartite resource states, and constructing these requires $2n$ noiseless channels.
In our MCT protocol, we use one $n$-partite resource, which requires $n$ noiseless channels to construct.   Therefore we reduce the cost by 50\%.

BQST

\section{Topological Compression: Definitions\label{Sect:CompressionDef}}
If the diagrammatic representation of a 2-qudit transformation $T$ has a free through string\footnote{By ``through string,'' we mean a neutral string that passes from the $j^{\rm th}$ input to the  $j^{\rm th}$ output and that crosses no other string.} on the left, 
\be\label{CompT}
\raisebox{-.4cm}{
\tikz{
\foreach \x in {0,1,2,3}{
\draw (\x,0)--++(0,1);
}
\fill[white] (0.5,.25) rectangle (3.5,.75);
\draw[thick,dashed] (0.5,.25) rectangle (3.5,.75);
\node at (2,.5) {$T$};
}}\ .
\ee
then we consider such transformations as topologically  compressed on the first qudit.

\begin{proposition}\label{Prop:Z-compress}
For an $n$-qudit transformation $T$, the following conditions are equivalent:
\begin{itemize}
\item[(1)] The transformation $T$ is block diagonal on the first qudit.
\item[(2)]There are $(n-1)$-qudit transformations $T(\ell)$, $l\in \Z_{d}$, so that
	\be\label{Controlled Transformation}
	T=\sum\limits_{\ell=0}^{d-1} \ket{\ell}\bra{\ell}\otimes T(\ell)\;,
	\ee
	i.e., $T$ is a controlled transformation, where the first qudit is the controlled qudit.
\item[(3)] There are qudit transformations $T'(\ell)$, $l\in \Z_{d}$, so that
	\be
	T=\sum\limits_{\ell=0}^{d-1} Z^{\ell} \otimes T'(\ell)\;.
	\ee
\item[(4)] The transformation $T$ commutes with Pauli $Z$ on the first qudit.
\end{itemize}

\begin{proof}
Obviously $(1) \iff (2) \iff (3) \Rightarrow (4)$.
Since $Z$ has distinct eigenvalues on the diagonal, we have that $(4)\Rightarrow (2)$.
\end{proof}
\end{proposition}
A transformation $T$ is called $Z$-compressed on the first qudit
if one of the above condition holds.

\begin{definition}  
In general, we say that a transformation $T$ is $Z$-compressed on the $j^{\rm th}$-qudit if $T$ commutes with the action of Pauli $Z$ on the $j^{\rm th}$-qudit. 
Similarly we say $T$ is $X$ (or $Y$)-compressed on the $j^{th}$ qudit, if  it commutes with the action of Pauli $X$ (or $Y$) on the $j^{th}$ qudit.
\end{definition}

We can switch between the three compressed transformations using
$FXF^{-1}=Z$ and $GXG^{-1}=Y^{-1}$.

\begin{theorem}
A transformation $T$ has the representation \eqref{CompT} if and only if it is $X$-compressed on the first qudit.
\end{theorem}
  
\begin{proof}
Applying the conjugation of $F$ on the first qudit to Proposition \ref{Prop:Z-compress}, we have the following equivalent conditions:
\begin{itemize}
\item[(1)] The transformation $T$ is $X$-compressed on the first qudit.\item[(2)] There are $(n-1)$-qudit transformations $T'(\ell)$, $l\in \Z_{d}$, so that
	\be
	T=\sum\limits_{\ell=0}^{d-1} X^{\ell} \otimes T'(\ell)\;.
	\ee
\item[(3)] The transformation $T$ commutes with Pauli $X$ on the first qudit.
\end{itemize}

The transformation on the second qudit $I\otimes T'(\ell)$ is represented by 
\be
\raisebox{-.4cm}{
\tikz{
\foreach \x in {0,1,2,3}{
\draw (\x,0)--++(0,1);
}
\fill[white] (1.5,.25) rectangle (3.5,.75);
\draw[thick,dashed] (1.5,.25) rectangle (3.5,.75);
\node at (2.5,.5) {$T'(\ell)$};
}}\ .
\ee
By Jordan-Wigner transformation, the $X\otimes I$ is represented by 
\be
\raisebox{-.4cm}{
\tikz{
\foreach \x in {0,1,2,3}{
\draw (\x,0)--++(0,1);
}
\node at (.5,.5) {$1$};
\node at (1.5,.5) {$-1$};
\node at (2.5,.5) {$1$};
}}\ .
\ee
Therefore if $T$ is $X$-compressed, then $T$ has the representation  \eqref{CompT} by condition (2).

On the other hand,  if $T$ has the representation  \eqref{CompT}, then it is algebraically generated by the three transfomations
\be%\label{CompT}
\raisebox{-.4cm}{
\tikz{
\foreach \x in {0,1,2,3}{
\draw (\x,0)--++(0,1);
}
\node at (.5,.5) {$1$};
}}\ 
\;,\qquad
\raisebox{-.4cm}{
\tikz{
\foreach \x in {0,1,2,3}{
\draw (\x,0)--++(0,1);
}
\node at (1.5,.5) {$1$};
}}
\;, \qquad
\raisebox{-.4cm}{
\tikz{
\foreach \x in {0,1,2,3}{
\draw (\x,0)--++(0,1);
}
\node at (2.5,.5) {$1$};
}}\;.
\ee
By para-isotopy, the three generators commutes with $X\otimes I$. So $T$ also commutes with $X\otimes I$. By condition (3), $T$ is $X$-compressed on the first qudit.
\end{proof}

The characterization and the proof also work for the $n$-qudit case.

\begin{definition}\label{Def:Compress}
The transformation $T'$ is compressed on the $j^{\rm th}$-qudit if $T'=UTV$, where $T$ is $Z$-compressed on the $j^{\rm th}$-qudit and  $U,V$ are local transformations on the $j^{\rm th}$-qudit.
\end{definition}

\section{The Mutipartite Compressed Teleportation (MCT) Protocol}\label{Sect:MCT}
\subsection{MCT for Controlled Transformations}\label{Sect:MCT-Control}
Suppose a network has one leader and $n$ parties. Also assume that the $j^{\rm th}$ party can perform a controlled transformation
\be
T_j=\sum_{l=0}^{d-1} \ket{\ell}\bra{\ell} \otimes T_{j}(\ell),
\ee
where the control qudit belongs to the person $P_j$ in the $j^{\rm th}$ party, and $T_j(\ell)$ can be an arbitrary multi-person, multi-qudit transformation on the targets. (The algebraic notation for  the controlled transformation $T_j$ is shown in Figure~\ref{Fig:Controlled T}.)
\begin{figure}[h]
\scalebox{0.8}{
\Qcircuit @C=1.5em @R=2em  {
& \gate{T_j(\ell)} \qwx[1] \qw & \qw\\
& \control \qw & \qw
}}
\caption{Controlled transformations. \label{Fig:Controlled T}}
\end{figure}

With these assumptions, we design a circuit shown on the left of Equation \eqref{MCT-Protocol}, 
where the resource state $\GHZ$ is represented by $d^{-\frac{1}{2}}\sum\limits_{k=0}^{d-1}\ket{k,k,\cdots,k}$.
We call it the Multipartite Compressed Teleportation protocol (MCT).

The function of the circuit is shown on the right of Equation \eqref{MCT-Protocol}.
The leader can perform any non-local controlled transformation $T_{j}$ to the $j^{\rm th}$ party  in the network, for $1\leqslant j\leqslant n$. 
The leader has the common control qudit, and the $j^{\rm th}$ party performs the transformation $T_{j}(\ell)$ for control qudit~$\ell$.

Surprisingly, one can implement these $n$ non-local transformations  using only $1$ resource state shared by the leader and the persons $P_j$.

\begin{theorem}
It costs  one $(n+1)$-partite resource state $\GHZ$ and $2n$ cdits (classical information channels) to implement controlled transformations $T_{j}$ shared by a Leader and the $j^{\rm th}$ party, for $1\leq j \leq n$. The time cost is the transmission of two cdits and the implementation of local transformations.
\end{theorem}

\vskip -2cm
\begin{align}\label{MCT-Protocol}
d^{-\frac{1}{2}}\sum_{k\in \Z_{d}}\quad\quad&\quad\quad\quad
\scalebox{0.8}{\raisebox{3cm}{
\Qcircuit @C=1.5em @R=2em  {
\lstick{\text{Party 1}}&\qw&\qw&\qw& \qw&\multigate{0}{T_{1}(\ell)}&\qw&\qw&\qw \\
\lstick{k}& \dstick{ \Large{\vdots}} \qw & \qw & \qw & \gate{X} \cwx[2] & \control \qw \qwx[-1]  & \gate{F^{-1}}& \meter \cwx[2]\\
\lstick{\text{Party n}}&\qw&\qw&\qw&\qw&\gate{T_{n}(\ell)}&\qw&\qw&\qw\\
\lstick{k}& \qw & \qw & \qw & \gate{X} \cwx[1] & \control \qw \qwx[-1] & \gate{F^{-1}} & \meter \cwx[2]\\
\lstick{k}&\gate{F^{-1}} & \control \qw \qwx[1] & \gate{F^{-1}}    & \meter \cwx[-1]  \\
\lstick{\text{Leader}}& \qw & \gate{Z} & \qw   & \qw & \qw & \qw & \gate{Z} & \qw\\
}}}\nonumber\\
&\nonumber\\
=\quad\quad&\quad\quad\quad
\scalebox{0.8}{\raisebox{1cm}{
\Qcircuit @C=1.5em @R=2em  {
\lstick{\text{Party 1}}& \multigate{0}{T_{1}(\ell)}&\qw &\qw &\qw\\
\lstick{\text{Party n}}& \qw& \ustick{ \Large{\vdots}} \qw&\multigate{0}{T_{n}(\ell)}&\qw \\
\lstick{\text{Leader}}& \control \qw \qwx[-2] & \qw& \control \qw \qwx[-1] & \qw
}}}
\end{align}

We specify the MCT in the usual algebraic terminology as a circuit. One can consider CT as a special case for two parties. 
From this picture, one can understand the protocol without knowing its topological significance.  In \S \ref{Sec:Topo MTC} we derive this protocol from topological simulation.

\section{Topological simulation for MCT}\label{Sec:Topo MTC}
We give the MCT diagrammatic protocol for $X$-compressed transformations. The design of this protocol is equivalent to the design for $Z$-compressed transformations by applying unitary conjugation. 
 
We summarize the use of topological simulation to design the protocol in Equation~\eqref{CTX-Protocol-Pic}.  Let us call the left-hand side of the identity ``Picture 1,''  the middle term in the identity ``Picture 2,'' and the right-hand side of the identity ``Picture 3.''  Picture 1 represents the simulation of the goal, 
where the Leader desires to share an $X$-compressed transformation $T_{j}$ with the $j^{\rm th}$ party, for $1\leqslant j\leqslant n$.

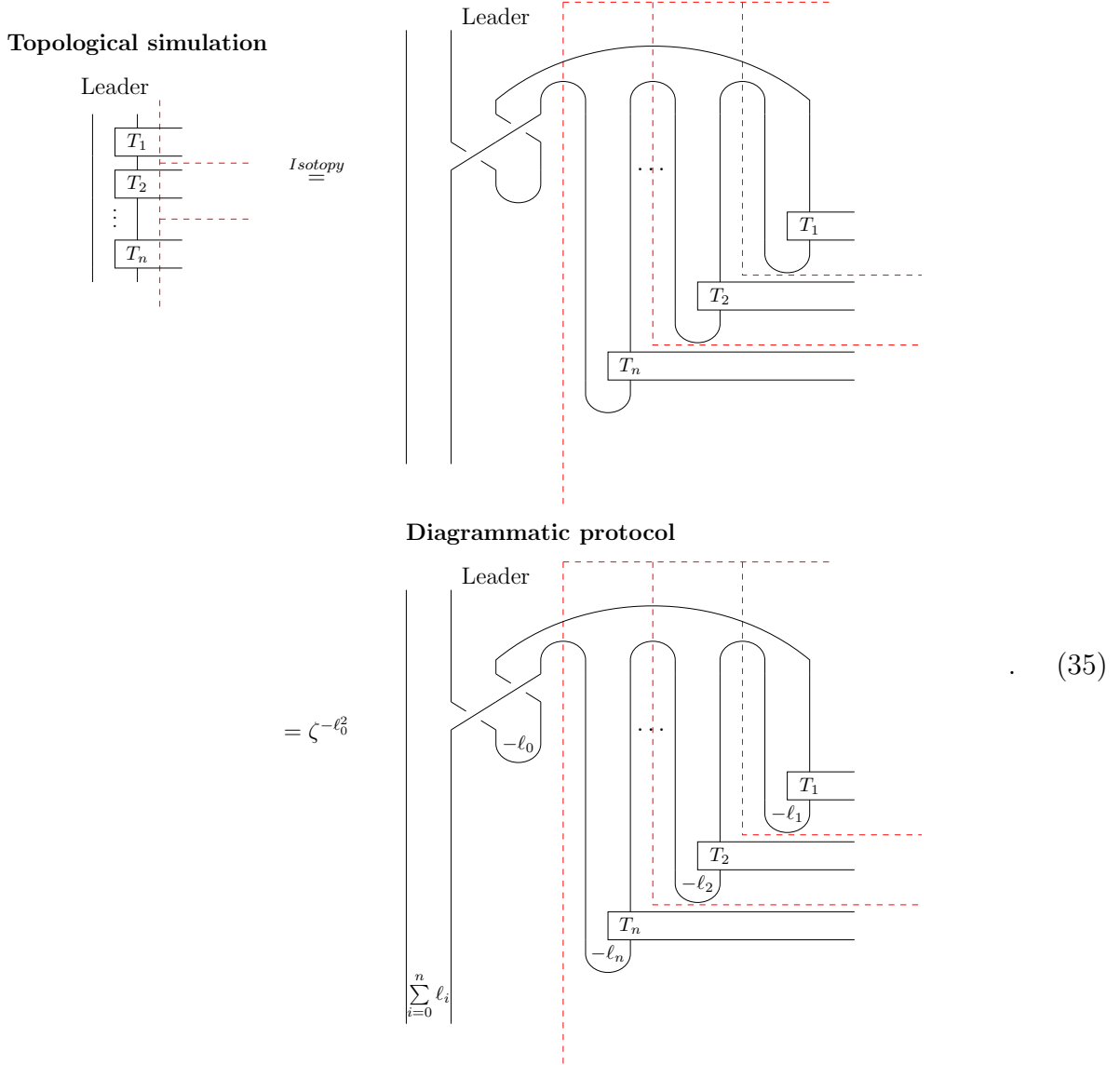
\be \label{CTX-Protocol-Pic}
\scalebox{0.8}{\raisebox{-7 cm}{
\begin{tikzpicture}
\begin{scope}[shift={(0,2)},xscale=1.2]
\node at (20,2) {$~$};
\node at (20,-9) {$~$};
\node at (21/3,-2/4) {Leader};
\draw[red,dashed] (23/3,-3/4)-- ++(0/3,-15/4);
\draw[red,dashed] (23/3,-7.5/4)--++(4/3,0);
\draw[red,dashed] (23/3,-11.5/4)--++(4/3,0);
\node at (22/3,1/4) {\textbf{Topological simulation}};
\draw (22/3,-5/4) -- (22/3,-4/4);
\draw (20/3,-7/4) -- (20/3,-4/4);
\draw (24/3,-5/4)--(21/3,-5/4) -- (21/3,-7/4) --(24/3,-7/4);
\node at (22/3,-6/4) {\size{$T_1$}};
\draw (22/3,-5/4-3/4) -- (22/3,-4/4-3/4);
\draw (20/3,-7/4-3/4) -- (20/3,-4/4-3/4);
\draw (24/3,-5/4-3/4)--(21/3,-5/4-3/4) -- (21/3,-7/4-3/4) --(24/3,-7/4-3/4);
\node at (22/3,-6/4-3/4) {\size{$T_2$}};
\node at (21/3,-11/4) {\size{$\vdots$}};
\draw (22/3,-5/4-8/4) -- (22/3,-4/4-6/4);
\draw (20/3,-7/4-8/4) -- (20/3,-4/4-6/4);
\draw (22/3,-5/4-11/4) -- (22/3,-4/4-11/4);
\draw (20/3,-5/4-11/4) -- (20/3,-4/4-11/4);
\draw (24/3,-5/4-8/4)--(21/3,-5/4-8/4) -- (21/3,-7/4-8/4) --(24/3,-7/4-8/4);
\node at (22/3,-6/4-8/4) {\size{$T_n$}};
\end{scope}
\node at (12,0) {$\overset{Isotopy}{=} $};
\begin{scope}[shift={(12,0)},xscale=1.2]
\node at (8/3,11/4) {Leader};
\node at (10/3,14/4) {\textbf{}};
\draw[red,dashed] (11/3,12/4)--++(4,0);
\draw[red,dashed] (11/3,12/4)-- ++(0/3,-36/4);
\draw[red,dashed] (15/3,12/4)-- ++(0/3,-24.5/4)--++(12/3,0);
\draw[red,dashed] (19/3,12/4)-- ++(0/3,-19.5/4)--++(8/3,0);
\fbraid{6/3}{0}{8/3}{2/4}
\fbraid{10/3}{4/4}{8/3}{2/4}
%%%%
\draw (8/3,4/4) -- (8/3,5/4) to [bend left=45] (22/3,5/4)--(22/3,4/4);
\fqudit{18/3}{4/4}{-1/3}{1/4}{}{}
\fqudit{14/3}{4/4}{-1/3}{1/4}{}{}
\fqudit{10/3}{4/4}{-1/3}{1/4}{}{}
\fmeasure{8/3}{0/4}{-1/3}{1/4}{}
\draw (10/3,0/4) -- (10/3,2/4);
\node at (15/3,0/4) {$\cdots$};
%%%%
\draw (22/3,-5/4+2/4) -- (22/3,4/4);
\draw (20/3,-7/4+2/4) -- (20/3,4/4);
\draw (24/3,-5/4+2/4)--(21/3,-5/4+2/4) -- (21/3,-7/4+2/4) --(24/3,-7/4+2/4);
\node at (22/3,-6/4+2/4) {\size{$T_1$}};
\fmeasure{20/3}{-7/4+2/4}{-1/3}{1/4}{}
%%%%
\draw (-4/3+22/3,-5/4-3/4) -- (-4/3+22/3,4/4);
\draw (-4/3+20/3,-7/4-3/4) -- (-4/3+20/3,4/4);
\draw (24/3,-5/4-3/4)--(-4/3+21/3,-5/4-3/4) -- (-4/3+21/3,-7/4-3/4) --(24/3,-7/4-3/4);
\node at (-4/3+22/3,-6/4-3/4) {\size{$T_2$}};
\fmeasure{-4/3+20/3}{-7/4-3/4}{-1/3}{1/4}{}
%%%%
\draw (-8/3+22/3,-5/4-8/4) -- (-8/3+22/3,4/4);
\draw (-8/3+20/3,-7/4-8/4) -- (-8/3+20/3,4/4);
\draw (24/3,-5/4-8/4)--(-8/3+21/3,-5/4-8/4) -- (-8/3+21/3,-7/4-8/4) --(24/3,-7/4-8/4);
\node at (-8/3+22/3,-6/4-8/4) {\size{$T_n$}};
\fmeasure{-8/3+20/3}{-7/4-8/4}{-1/3}{1/4}{}
%%%%
\draw (4/3,10/4) -- (4/3,-21/4);
\draw (6/3,0) -- (6/3,-21/4);
\draw (6/3,10/4) -- (6/3,2/4);
\end{scope}

\node at (12,-10) {$\overset{}{=} \zeta^{-\ell_{0}^{2}}$};
%\end{scope}
\begin{scope}[shift={(12,-10)},xscale=1.2]
\node at (8/3,11/4) {Leader};
\node at (10/3,14/4) {\textbf{Diagrammatic protocol}};
\draw[red,dashed] (11/3,12/4)-- ++(4,0);
\draw[red,dashed] (11/3,12/4)-- ++(0/3,-36/4);
\draw[red,dashed] (15/3,12/4)-- ++(0/3,-24.5/4)--++(12/3,0);
\draw[red,dashed] (19/3,12/4)-- ++(0/3,-19.5/4)--++(8/3,0);
\fbraid{6/3}{0}{8/3}{2/4}
\fbraid{10/3}{4/4}{8/3}{2/4}
%%%%
\draw (8/3,4/4) -- (8/3,5/4) to [bend left=45] (22/3,5/4)--(22/3,4/4);
\fqudit{18/3}{4/4}{-1/3}{1/4}{}{}
\fqudit{14/3}{4/4}{-1/3}{1/4}{}{}
\fqudit{10/3}{4/4}{-1/3}{1/4}{}{}
\fmeasure{8/3}{0/4}{-1/3}{1/4}{}
\node at (9/3,-1/4) {\size{$-\ell_0$}};
\draw (10/3,0/4) -- (10/3,2/4);
\node at (15/3,0/4) {$\cdots$};
%%%%
\draw (22/3,-5/4+2/4) -- (22/3,4/4);
\draw (20/3,-7/4+2/4) -- (20/3,4/4);
\draw (24/3,-5/4+2/4)--(21/3,-5/4+2/4) -- (21/3,-7/4+2/4) --(24/3,-7/4+2/4);
\node at (22/3,-6/4+2/4) {\size{$T_1$}};
\fmeasure{20/3}{-7/4+2/4}{-1/3}{1/4}{-\ell_1}
%%%%
\draw (-4/3+22/3,-5/4-3/4) -- (-4/3+22/3,4/4);
\draw (-4/3+20/3,-7/4-3/4) -- (-4/3+20/3,4/4);
\draw (24/3,-5/4-3/4)--(-4/3+21/3,-5/4-3/4) -- (-4/3+21/3,-7/4-3/4) --(24/3,-7/4-3/4);
\node at (-4/3+22/3,-6/4-3/4) {\size{$T_2$}};
\fmeasure{-4/3+20/3}{-7/4-3/4}{-1/3}{1/4}{-\ell_2}
%%%%
\draw (-8/3+22/3,-5/4-8/4) -- (-8/3+22/3,4/4);
\draw (-8/3+20/3,-7/4-8/4) -- (-8/3+20/3,4/4);
\draw (24/3,-5/4-8/4)--(-8/3+21/3,-5/4-8/4) -- (-8/3+21/3,-7/4-8/4) --(24/3,-7/4-8/4);
\node at (-8/3+22/3,-6/4-8/4) {\size{$T_n$}};
\fmeasure{-8/3+20/3}{-7/4-8/4}{-1/3}{1/4}{-\ell_n}
%%%%
\node at (5/3,-19/4) {\size{$\sum\limits_{i=0}^n \ell_i$}};
\draw (4/3,10/4) -- (4/3,-21/4);
\draw (6/3,0) -- (6/3,-21/4);
\draw (6/3,10/4) -- (6/3,2/4);
\end{scope}
\end{tikzpicture}

}}.
\ee

The non-local transformation $T_{j}$ cannot be implemented directly. We first apply topological isotopy, in a way that isolates each transformation $T_{j}$ in the region of the $j^{\rm th}$ party. (These regions are separated by the red dashed lines). Then each $T_{j}$ becomes local. We also move the intersections of the strings and the red dashed lines to the top, so that the non-locality only appears in the state, which turns out to be the resource state. 
This is how one obtains Picture 2 from Picture 1.

The cups in Picture 2 of  Equation~\eqref{CTX-Protocol-Pic} represent measurements. 
We add charges on cups to indicate the results of the measurements. Each resulting charge in a measurement must be balanced by an opposite charge. We  add that charge on the corresponding string in the region of the Leader.  This may also give  a  global phase  from applying the string Fourier relation~\eqref{Equ:SF1} for the Leader's charge.  These charged strings define measurement-based recovery maps given by Pauli $X$.   Thus we arrive at Picture 3, which is a diagrammatic protocol for MTC. It includes one multipartite resource state $\Max$ and LOCC.

We construct the diagrammatic MCT protocol using the 2-string language. From the above topological simulation, we observe that a natural resource state for multipartite communication is $\Max$, which we recognize from \eqref{Max-State}.   In fact, $\Max$ is unitary equivalent to the usual resource state $\GHZ$. The measurement-based recovery map arising from topological design is given by Pauli matrices.  This is a general phenomenon in various protocols for communication.

Using the two-string language dictionary in~\cite{JLW}, one can translate this diagrammatic protocol to the algebraic circuit given  in \eqref{Algebraic-MCTX}.

\bigskip
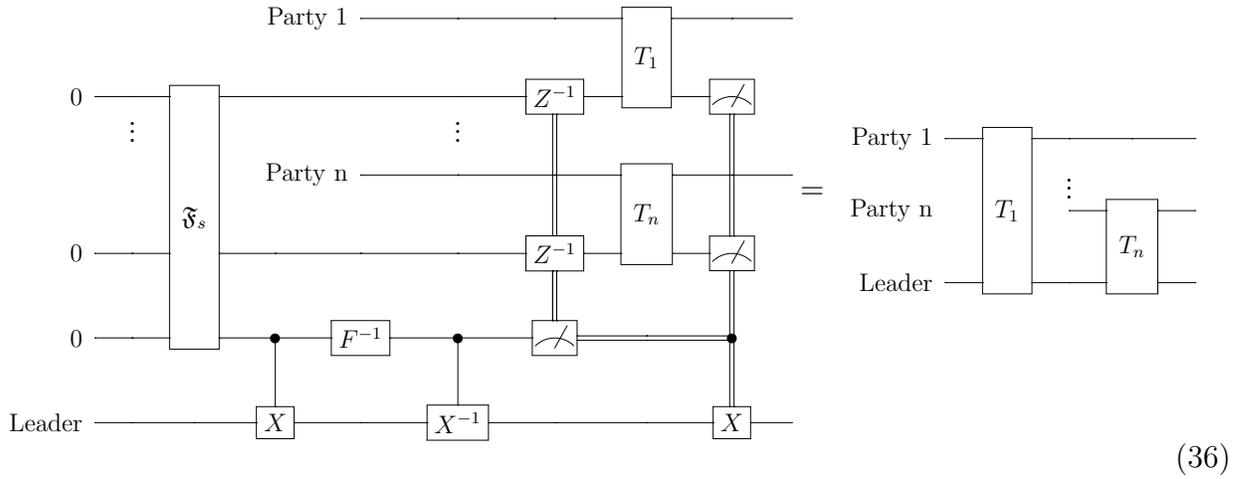
\begin{align}\label{Algebraic-MCTX}
\qquad
\scalebox{0.8}{\raisebox{3cm}{
\Qcircuit @C=1.5em @R=2em  {
&&&&\lstick{\text{Party 1}}&\qw&\qw& \multigate{1}{T_1}&\qw&\qw\\
\lstick{0}   & \dstick{ \Large{\vdots}} \qw                 & \multigate{3}{\FS} & \qw & \qw & \dstick{ \Large{\vdots}} \qw & \gate{Z^{-1}} \cwx[2] & \ghost{T_1}  & \meter \cwx[2]\\
&&&&\lstick{\text{Party n}}&\qw&\qw&\multigate{1}{T_n}&\qw&\qw\\
\lstick{0}  &  \qw &  \ghost{\FS}       & \qw & \qw & \qw & \gate{Z^{-1}} \cwx[1] & \ghost{T_n}  & \meter \cwx[1]\\
\lstick{0} & \qw                     & \ghost{\FS}  & \control \qw \qwx[1] & \gate{F^{-1}} & \control \qw \qwx[1]   & \meter \cwx[-1] & \cw & \control \cw \cwx[1] \\
\lstick{\text{Leader}}& \qw & \qw  & \gate{X} & \qw & \gate{X^{-1}} & \qw & \qw & \gate{X} & \qw\\
}}}
&=\qquad\quad\quad~~\scalebox{0.8}{\raisebox{1cm}{
\Qcircuit @C=1.5em @R=2em  {
\lstick{\text{Party 1}}& \multigate{2}{T_1}&\qw &\qw &\qw\\
\lstick{\text{Party n}}  & & \ustick{ \Large{\vdots}}    &\multigate{1}{T_n}&\qw \\
\lstick{\text{Leader}}& \ghost{T_1} & \qw  & \ghost{T_n} & \qw
}}}\\ \nonumber
&
\end{align}
Here one simplifies the protocol by the identity in Figure~\ref{Trick}.
\begin{figure}[h]
\scalebox{0.8}{
\Qcircuit @C=1.5em @R=2em  {
& \control \qw \qwx[1]  & \meter \cwx[1] \\
& \gate{X^{-1}}         & \gate{X} & \qw\\
}}
\raisebox{-.5cm}{=}
\scalebox{0.8}{
\Qcircuit @C=1.5em @R=2em  {
&  \qw  & \meter  \\
& \qw   & \qw & \qw\\
}}
  \caption{\label{Trick}}
\end{figure}

We represent $\Max$ as $\FS |\vec {0}\rangle$ in Equation \eqref{Algebraic-MCTX}, where
the state $|\vec{0}\rangle=\ket{0,0,\ldots,0}$ denotes the $n$-qudit with charge $0$ for each $1$-qudit, and we call $|\vec 0\rangle$ the ground state. 
We mention the extremely interesting transformation $\FS$ that appears here, and that we call the {\it string Fourier transform}.  
It is a mechanism to produce the $n$-qudit resource state $\Max$ from the ground state $|\vec {0}\rangle$.
We explore $\FS$  extensively in \cite{JL}.

Taking the conjugation of
local transformations, we obtain the MCT protocol for other types of compressed transformations. In particular, taking the conjugate of
the Fourier transform $F$, we obtain the MCT protocol for $Z$-compressed transformations or controlled transformations in \eqref{MCT-Protocol} .

In the case with only two persons, the MCT protocol says:
Assume that a quantum network can perform a transformation $T$, which is compressed on a 1-qudit belonging to a network member Alice. Then Alice can teleport her 1-qudit transformation to Bob using one edit and two cdits.
One can easily derive the entanglement-swapping protocol, and the teleportation of the Tofolli gate from it.

\section{Conclusion}
In this paper we extend our  two-string model for quantum information.  
\begin{enumerate}
\item{} We articulate the concept of constructive simulation and topological design.   

\item{} We  introduce  topological compression and define compressed transformations.

\item{} We define a protocol to teleport compressed transformations. 

\item{} Our new protocol costs only one multipartite resource state to implement multiple, non-local transformations between multiple parties. For more than two parties, our multipartite teleportation protocol does not reduce to compositions of bipartite communications.
\end{enumerate}

\section{Acknowledgements}
\noindent 
We are grateful to Jacob~Biamonte, Bob~Coecke, Michael Freedman,  Amar~Hadzihasanovic, David~Reutter, and Kevin Walker for interesting discussions during an October 2016 workshop at Harvard University, and Jamie Vicary for correspondence.
We thank  our hosts for hospitality at the Research Institute for Mathematics of the ETH-Zurich,   at the Max Planck Institute for Mathematics, and at the Hausdorff Institute for Mathematics in Bonn, where part of this work was carried out.  This research was supported in part by a grant from the Templeton Religion Trust.

\end{document}